%% file: main.tex
\documentclass[conference]{IEEEtran}
\usepackage{blindtext, graphicx}
\usepackage{amsmath,amsfonts,amssymb,bbm}
\usepackage{algorithm,algorithmic}
\usepackage{epstopdf}
\usepackage{epsfig}
\graphicspath{{images/}}
\usepackage{mathtools,dsfont}
\usepackage{psfrag,bbm,graphicx}
\usepackage{algorithm,algorithmic}
\usepackage{comment,balance}
\usepackage{epstopdf}
\usepackage{epsfig}
\usepackage{multicol}
\usepackage{multirow, cite}
\usepackage{amsfonts}
\usepackage{color}

\def \OO {\mathrm{O}}
\def \oo {\mathrm{o}}
\def\EE{{\mathbb{E}}}
\def\PP{{\mathbb{P}}}

\ifCLASSINFOpdf
  
\else
 
\fi

\newcounter{lecnum}

\newtheorem{theorem}{Theorem}
\newtheorem{assumption}{Assumption}
\newtheorem{lemma}{Lemma}
\newtheorem{proposition}{Proposition}

\newtheorem{corollary}{Corollary}

\newtheorem{remark}{Remark}
\newenvironment{proof}{{\bf Proof:}}{\hfill\rule{2mm}{2mm}}
\newcommand{\remove}[1]{}
\begin{document}
%
\title{Effects of Storage Heterogeneity \\in Distributed Cache Systems}

\author{Kota Srinivas Reddy, Sharayu Moharir and Nikhil Karamchandani \\
	Department of Electrical Engineering, Indian Institute of Technology Bombay \\
	Email: ksvr1532@gmail.com, sharayum@ee.iitb.ac.in, nikhilk@ee.iitb.ac.in
}

\maketitle

\begin{abstract}
	In this work, we focus on distributed cache systems with non-uniform storage capacity across caches. We compare the performance of our system with the performance of a system with the same cumulative storage distributed evenly across the caches. We characterize the extent to which the performance of the distributed cache system deteriorates due to storage heterogeneity. The key takeaway from this work is that the effects of heterogeneity in the storage capabilities depend heavily on the popularity profile of the contents being cached and delivered. We analytically show that compared to the case where contents popularity is comparable across contents, lopsided popularity profiles are more tolerant to heterogeneity in storage capabilities. We validate our theoretical results via simulations. 
%
%
\end{abstract}


%
\IEEEpeerreviewmaketitle

\input{intro.tex}
\input{setting.tex}

\input{mainresultslessthanone.tex}

\input{mainresultsmorethanone.tex}
\input{simulations.tex}

\input{proofs2.tex}
\bibliographystyle{IEEEtran}
\bibliography{myref2}





\end{document}

%% file: intro.tex
\section{Introduction}\label{sec:introduction}

{\let\thefootnote\relax\footnote{This work was supported in part by the Bharti Centre for Communication at IIT Bombay. The work of Sharayu Moharir and Nikhil Karamchandani was supported in part by seed grants from IIT Bombay and an Indo-French grant on ``Machine Learning for Network Analytics". The work of Nikhil Karamchandani was also supported in part by the INSPIRE Faculty Fellowship from the Govt. of India.}}
Recent Internet usage patterns show that Video on Demand (VoD) services, e.g., YouTube \cite{Youtube} and Netflix \cite{Netflix}, account for ever-increasing fractions of Internet traffic \cite{Cisco}. To meet the increasing demand, most popular VoD services use content delivery networks (CDNs). We focus on multiple geographically co-located caches, each with limited storage and service capabilities, deployed to serve users in that area. The motivation behind deploying local caches is to serve most user requests locally. Requests that can't be served locally are served by a central server (which stores the entire content catalog) via a root node, see Figure \ref{fig:cache_cluster}. This setting, also studied in \cite{moharir2016content}, models networks where, \emph{(i)} the ISP (root node) uses local caches to reduce the load on the network backbone or \emph{(ii)} this geographically co-located cache cluster is a part of a larger tree network \cite{borst2010distributed}.

Most VoD service offer catalogs consisting of a large number of contents and serve a large number of users. Motivated by this, we study a time-slotted setting where a batch of requests arrive in each time-slot and every cache can serve at most one request in a batch. Requests that cannot be served locally by the caches are assigned to the central server. Storage and service policies are designed to minimize the number of contents which need to be fetched from the central server to serve all the requests in a batch.



\begin{figure}[t]
	\begin{center}
		\includegraphics[scale=0.33]{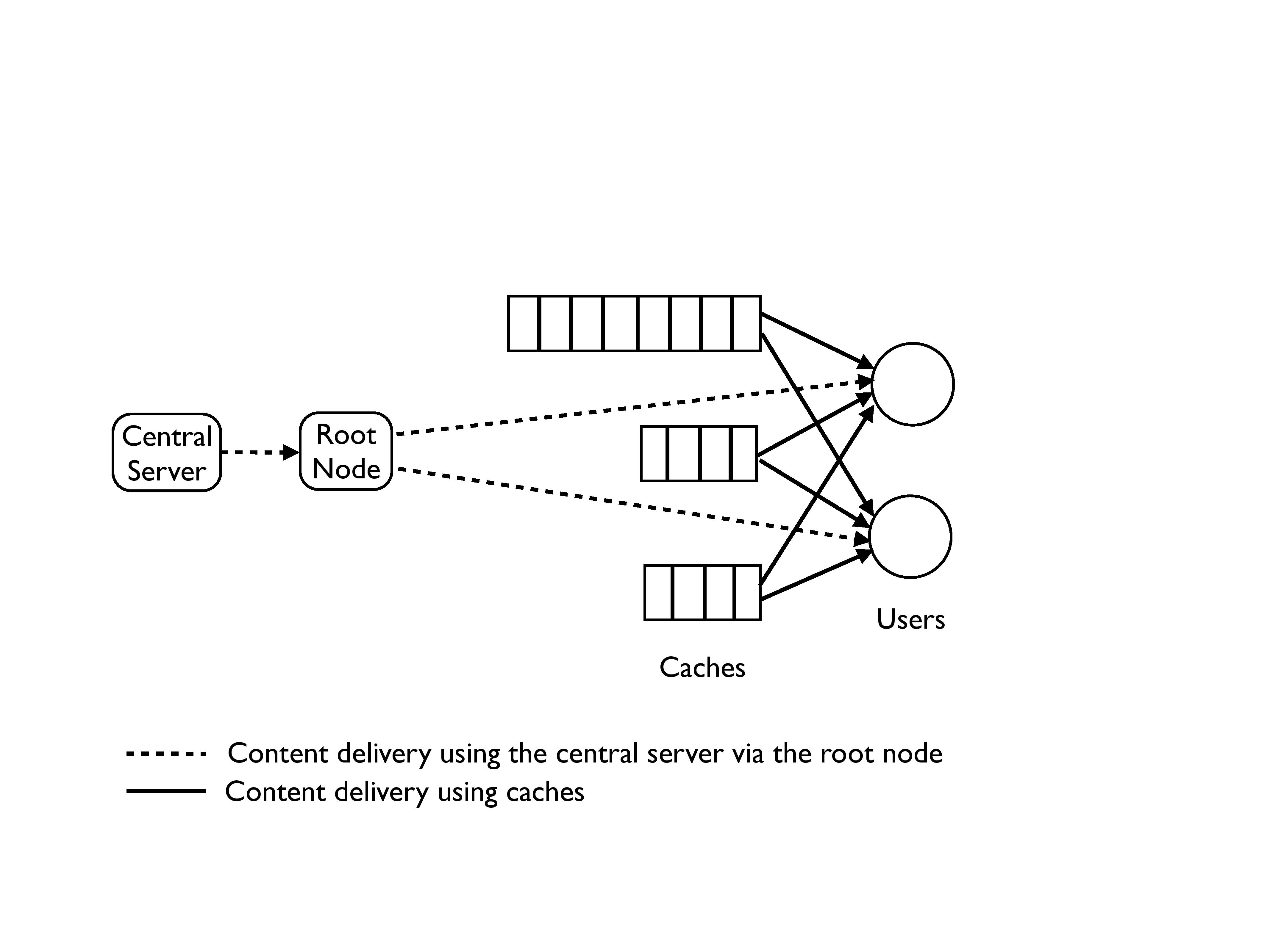}
		\caption{\sl An illustration of a cache cluster consisting of three caches serving two users. The first cache has more storage than the other two. Each user can either be served by the caches or by the central server via the root node. \label{fig:cache_cluster}}
	\end{center}
\end{figure}

The existing body of work in this space considers the setting where storage capabilities are uniform across caches \cite{reddy2017resource, moharir2016content}. In this work, we study the effects of heterogeneity in storage across caches on the performance of the system. The key takeaway of this work is that the effect of heterogeneity in storage capabilities across caches depends on the popularity profile of the contents. 
We show that as content popularity becomes more lopsided, the  system can handle more heterogeneity in cache storage capabilities, i.e., for the same amount of cumulative memory, the performance of the heterogeneous system remains comparable to the performance of a system with uniform storage across caches. 

Intuitively our results can be explained as follows. Increasing the number of contents stored on a cache increases the utility of that cache as it can be used to serve a request for any one of the stored contents. When content popularity is comparable across contents, the fraction of requests in a batch for any particular content  is small. As a result, for a cache with limited storage, it is likely that none of the stored contents are requested, thus leaving the cache unutilized. This increases the number of requests that have to be served via the central server. In contrast, when content popularity is lopsided, the caches with limited storage can be used to store and serve requests for popular contents and the caches with large storage can store a mixture of some popular and a larger number of unpopular contents. This ensures that most caches are utilized, thus reducing the number of requests served centrally.

The main focus of this work is to study the impact of heterogeneity in storage sizes on the performance of a single-layer distributed caching system with a central server. This aspect has been addressed in some other settings as well. \cite{rossi2012sizing} models a caching network as a graph with a cache at each vertex and explores sizing the individual caches according to various vertex centrality metrics. 	\cite{abd2017cache} studies a multi-tier caching network, with a possibly different cache size at each layer. The setting where each user is pre-matched to a server and the central server communicates with the users via an error-free broadcast link has been studied recently under the moniker \textit{`coded caching'} in \cite{maddah2014fundamental} and the impact of heterogeneity in cache sizes in this setting has been explored in \cite{wang2015coded, ibrahim2017centralized, karamchandani2016hierarchical}.

\subsection{Contributions}
The main contributions of this work can be summarized as follows:
\begin{enumerate}
\item We first consider the case where the content popularity distribution follows the Zipf distribution (defined in Section~\ref{request}) with parameter less than $1$, which corresponds to the case where the popularity is comparable across contents. We show that even if  a constant fraction of the caches are restricted to have small memory size as compared to the remaining caches, the required expected server transmission  rate can be much larger than a homogeneous system with the same  cumulative memory.
\item Next, we consider the case where the content popularity distribution follows the Zipf distribution with parameter larger than $1$ and as a result, the content popularity is more lopsided. Unlike the previous case, even if all the memory is concentrated in only a vanishing fraction of the caches, the performance of the system will be similar to a homogeneous system with the same cumulative memory. 
\end{enumerate}
The above results suggest that caching systems are more tolerant to heterogeneity in storage under content popularity distributions which are more lopsided than when  popularity is comparable across contents.

%% file: setting.tex
\section{setting} \label{sec:setting}

We study a system consisting of a central server, and $m$ co-located caches, each with limited storage and service capabilities. The central server stores $n$ files\footnote{We use the terms `content' and `file' interchangeably.} of equal size (say $1$ unit = $b$ bits), where $n=\Theta(m^\gamma),$ for some $\gamma\geq 1$. Users make requests for these files, and the user requests are served using the caches and the central server. 

The system operates in two phases: the first phase is the \textit{placement phase,} in which each cache stores content related to the $n$ files and the next phase is the \textit{delivery phase,} in which a batch of requests arrives and are served by the caches and the central server. While files can be split for storage and transmission, this work is restricted to uncoded policies during the placement and the delivery phases. We study the asymptotic performance of the system as $n$, $m$ $\rightarrow \infty.$
	
\subsection{Storage Model} \label{storage}
  Cache $i$ has the capacity to store $k_i$ units of data. Let $M = \sum_{i=1}^{m}k_i$ denote the cumulative cache memory. Without loss of generality, we assume caches are arranged in decreasing order of storage capacity, i.e., if $i<j$, then $k_i\geq k_j$.

\subsection{Request Model} \label{request}
We assume a time-slotted system. In each time-slot, a batch of $\widetilde{m}=\rho m$ (for some $\rho<1$) requests arrive from users according to an i.i.d. distribution. Files are indexed in decreasing order of popularity. 

Numerous empirical studies have shown that content popularity in VoD services follows the Zipf's law \cite{liu2013measurement,BC99,YZ06,fricker2012impact}. Zipf's law states that the popularity of the $i^{\text{th}}$ most popular content is proportional to $i^{-\beta}$, where $\beta$ is a positive constant known as the Zipf parameter. Small values of $\beta$ imply that content popularity is comparable across contents while larger values of $\beta$ correspond to lopsided popularity distributions. Typical values of  $\beta$ lie between 0.6 and 2.

%
%
%

\subsection{Service Model} \label{service}
We assume a delay-intolerant uncoded service system, i.e., all user requests in a given time slot have to be served jointly by the caches and the central server in that time-slot without queuing and coding. To begin with, depending on user requests, we match users with the caches such that no cache is matched\footnote{The more general setting where each cache can serve upto $a \ge 1$ requests simultaneously has been analyzed in \cite{reddy2017resource} for the case of homogeneous cache sizes. A similar analysis can be attempted for the case of heterogeneous cache sizes, however we do not pursue that direction in this paper.} to more than $1$ user. 
Depending on the user requests and the matching between the users and the caches, the central server then transmits a message to the root node which then relays it directly to the users. 
 Using the data received from the assigned caches and the central server message, each user should be able to reconstruct the requested file. Refer to Figure \ref{fig:example} for an example. 

\begin{figure}[t]
	\begin{center}
		\includegraphics[scale=0.33]{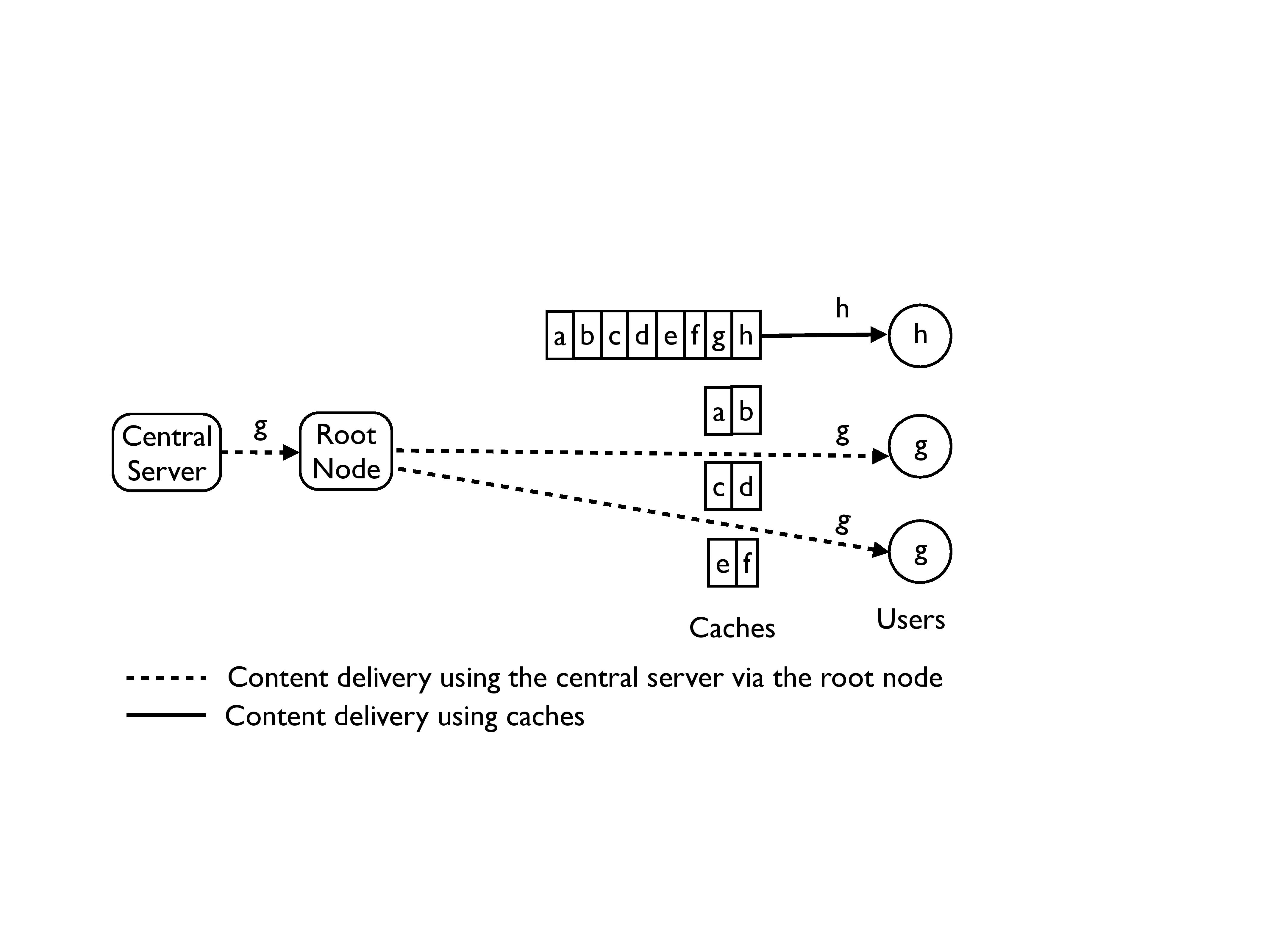}
		\caption{\sl An illustration of a system consisting of four caches serving three users. The catalog consists of eight files $\{a,b,c,d,e,f,g,h\}$. The first user requests file $h$, while the other two request file $g$. The first user is served by the first cache. The other two users are served by the central server. Since both users request for the same file $g$, the central server sends file $g$ to the root node, therefore, the transmission rate in this example is one. \label{fig:example}}
	\end{center}
\end{figure}

\subsection{Goal} \label{goal}
The reason for deploying local caches is that they can help reduce the load on the bottleneck link between the central server and the root node. The goal in such systems is to design efficient storage and service policies to reduce the expected transmission rate required from the central server to satisfy all user requests, where the expectation is with respect to the file popularity distribution. 
Note that if a content needs to be delivered by the central server to multiple users in a time-slot, the central server transmits it to the root node only once. Our storage and service policies depend on the file popularity distribution.

In a departure from the existing body of work on content caching/delivery policies, we characterize the performance of various caching policies for the setting where storage is heterogeneous across caches.

\remove{

1) I think we need a better figure, which illustrates the user requests coming in, their matching to various caches and the role of the central server sending a message to the root. 

2) It's not clear what kind of strategies we are allowing. I think we need to somehow mention here for example that we allow caches to store pieces of files, instead of full files. Similarly, requests can also broken down into requests for pieces etc. In addition, how the message of the central server reaches the end users is not clear. Is it directly to the users, or via the root node, or via the caches? Perhaps the figure can help with this as well, looking at the current figure I don't see any broadcasting and so if we use the term `broadcast' it might be confusing. 

3) Should we indicate somewhere that we are assuming that the contents of the cache remain static while new demands come in each slot. This models a slow changing catalogue and popularity distribution.

4) One example to illustrate storage and service policy and transmission rate.
}

%% file: mainresultslessthanone.tex
%

\section{Main results and discussion} \label{sec:results}
In this section, we state and discuss our main results. Proofs are given in Section \ref{sec:proofs}. 

We study distributed cache systems characterized as follows:

\begin{assumption}[Distributed Cache System]
	\label{ass:system}
	\begin{enumerate}
		\item[--] $m$ caches.
		\item[--] $n = m^{\gamma}$ files for $\gamma \geq 1$.
		\item[--] All files are of equal size, normalized to one unit.
		\item[--] File popularity: Zipf distribution with parameter $\beta$.
		\item[--] Cumulative cache memory is $M$ units, where $M = m^{\mu}$ for $\mu\geq 1$. 
		\item[--] Each cache can store at least one full file, i.e., $k_i \ge 1$ $\forall \ i$.
		\item[--] Requests are received in batches of $\widetilde{m} = \rho m$, where $\rho < 1.$ Each request is generated i.i.d. according to the popularity distribution.
		\item[--] At most one request in a batch can be allocated to each cache.
	\end{enumerate}
\end{assumption}

\subsection{Zipf distribution with $\beta \in [0,1)$}\label{zipf} 
\label{sec:lessthanone}
We first characterize the performance of a distributed cache system when file popularity follows the Zipf distribution with parameter $\beta \in [0,1)$. 

In addition to understanding the fundamental limit on the performance of any policy, we also evaluate the performance of a policy called Proportional Placement and Maximum Matching (PPMM) proposed in  
\cite{LLM12}.  \cite[Theorem 1]{reddy2017resource} characterizes the performance of the PPMM policy for a homogeneous cache system with the number of caches scaling linearly with the number of files. \color{black} In the PPMM policy, the number of caches that store copies of a file are proportional to its popularity. File copies are stored on caches such that no cache stores the same file multiple times. Once a batch of requests is revealed, a bipartite graph $G(V_1,V_2, E)$ is created, where $V_1$ is the set of requests, $V_2$ is the set of caches, and $E$ is the set of edges. There is an edge between  $v_1\in V_1$ and $v_2 \in V_2$ if Cache $v_2$ can serve request $v_1$, $i.e.$, if it stores a copy of the requested file. Once the bipartite graph is created, the maximum cardinality matching between the set of requests ($V_1$) and the set of caches ($V_2$) is found. All the matched requests are served by the corresponding caches and all the unmatched requests are served by the central server via the root-node. Note that this policy satisfies our service constraint that no cache is allocated more than one request in a batch.

The following result is a straightforward generalization of \cite[Theorem 1]{reddy2017resource} and characterizes the performance of the PPMM policy for a homogeneous cache system, i.e., a system where all caches have the same storage capabilities.

%
%
%
%

\begin{theorem} \label{thm:zipf01}
	Consider a homogeneous distributed cache system satisfying Assumption \ref{ass:system} where all caches have equal storage capacity of $M/m$ units and file popularity follows the Zipf distribution with parameter $\beta \in [0,1)$. For this system, let $R^{\text{PPMM}}_{z_{[0,1)}}$ be the central server's transmission rate for the PPMM policy described above. Then, we have that,
		\begin{eqnarray*}
			\mathbb{E}\left[R^{\text{PPMM}}_{z_{[0,1)}}\right]
			= \begin{cases}
				\OO(m)& \text{if } M <(1-\epsilon) n, \text{ } {\epsilon>0,} \\
				\OO\big( m^2 e^{-\frac{M}{n}}\big) &\text{if } M \geq  n. 
			\end{cases}
		\end{eqnarray*}	


\end{theorem}

We use this result to characterize the amount of memory needed in a homogeneous system to ensure that for the PPMM policy, the expected transmission rate of the central server goes to zero as the system size $m$ scales. 

%
%
%
%

\begin{corollary}\label{cor:zipf_0to1}
	Consider a homogeneous distributed cache system satisfying Assumption \ref{ass:system} where all caches have equal storage capacity of $M/m$ units and file popularity follows the Zipf distribution with parameter $\beta \in [0,1)$. For this system, let $R^{\text{PPMM}}_{z_{[0,1)}}$ be the central server's transmission rate for the PPMM policy described above. 
		If $M \geq {3n\ln m}= \Omega\big(n\ln m\big)$, $\EE\big[R^{\text{PPMM}}_{z_{[0,1)}}\big]=\oo(1)$. 
\end{corollary}

We thus conclude that, for a homogeneous distributed cache system and the PPMM policy, a cumulative cache memory of $M = \Omega(n \ln m)$ is sufficient to ensure that the expected transmission rate of the central server goes to zero as the system size $m$ scales. 


Our next result focuses on a heterogeneous distributed cache system, i.e., a distributed cache system where storage is non-uniform across caches. It characterizes the fundamental limit on the performance of any policy and evaluates the performance of the PPMM policy for such a system.


\begin{theorem}\label{thm:het01}
		Consider a heterogeneous distributed cache system satisfying Assumption \ref{ass:system} where file popularity follows the Zipf distribution with parameter $\beta \in [0,1)$. 
		\begin{enumerate}			
			\item[(a)] \emph{Lower bound on transmission rate}: Let $\widetilde{R}^*_{z_{[0,1)}}$ be the central server's transmission rate for optimal policy. There exists an {$\alpha (\rho)  \in (0,1)$} such that if a fraction of the $m$ caches, say $m_2= \alpha \cdot m$ caches, have at most $\OO\big(n / m^{\frac{1}{1-\beta}}\big) $ units of memory, then,
			$\EE\big[\widetilde{R}^*_{z_{[0,1)}}\big]=\omega(1).$
			\item[(b)] \emph{Performance of PPMM}: Let $\widetilde{R}_{z_{[0,1)}}^{\text{PPMM}}$ be the central server's transmission rate for the PPMM policy described above. Then for any $c>0$ and $\delta<1$, if a fraction of the $m$ caches, say $m_1=(\rho+c)m$ caches, have at least $\Omega(n/m^{\delta})$ units of memory, then, $\EE\big[\widetilde{R}^{\text{PPMM}}_{z_{[0,1)}}\big]=\oo(1).$
		\end{enumerate}		
\end{theorem}

Intuitively, the lower bound on the transmission rate for any policy can be explained as follows. Increasing the number of contents stored on a cache increases the potential utility of that cache as it can be used to serve a request for any one of the stored contents. When storage is non-uniform across caches and content popularity is comparable across files, the utility of caches with limited storage capabilities is small since content popularity being comparable across files ensures that the fraction of requests for any particular content is small. A consequence of this is that it is very likely that many caches with limited memory go unutilized when serving a batch of requests. A large number of unutilized caches is equivalent to a large number of requests being served via the central server, thus increasing its transmission rate.

In the next result, we use Theorem \ref{thm:het01}(a) and Corollary \ref{cor:zipf_0to1} to highlight the difference between homogeneous and heterogeneous cache systems with the same cumulative memory.

\begin{corollary}
\label{Cor:HetHom}
Let $\alpha \in (0,1)$ be as defined in Theorem~\ref{thm:het01}(a) and $\delta > 0$. Consider distributed caching systems satisfying Assumption \ref{ass:system} where file popularity follows the Zipf distribution with parameter $\beta \in [0,1)$. Consider a heterogeneous system such that $\alpha \cdot m$ caches each have at most $\OO\big(n / m^{\frac{1}{1-\beta}}\big)$ units of memory and the remaining $(1-\alpha)m$ caches each have memory $\Theta\big(n / m^{1 - \delta}\big)$. Then,
			$\EE\big[\widetilde{R}^*_{z_{[0,1)}}\big]=\omega(1).$
			
On the other hand, for a homogeneous system with the same cumulative cache memory $M = \Theta\big(n\cdot m^{\delta}\big)$ as the above heterogeneous system, we have 
$\EE\big[R^{\text{PPMM}}_{z_{[0,1)}}\big]=\oo(1).$
\end{corollary}

%% file: mainresultsmorethanone.tex
\subsection{Zipf distribution with $\beta >1$}
\label{section:morethanone}
We now compare the performances of homogeneous and  heterogeneous systems when file popularity follows the Zipf distribution with parameter $\beta >1$. The next result provides lower bounds on the expected transmission rate for a system with cumulative cache storage of $M$ units.

\begin{theorem} \label{thm:conv_het}		
Consider a distributed cache system satisfying Assumption \ref{ass:system} where file popularity follows the Zipf distribution with parameter $\beta> 1$. Let $\widetilde{R}_{z_{>1}}^*$ denote the optimal transmission rate for any uncoded storage/service policy. \\
-- If $\gamma \leq \dfrac{1}{\beta-1}$,
\begin{align*}
	 \EE\big[\widetilde{R}_{z_{>1}}^*\big]
	& = \begin{cases}
		\Omega\left(m^{1-\mu(\beta-1)}\right)&\text{if } M\le(1-\epsilon)n,\text{ } \epsilon>0, \\
		\Omega\left(m^{\frac{2-\mu\beta}{\beta}}\right)&\text{if } M=n, 
	\end{cases} \\
\EE\big[\widetilde{R}_{z_{>1}}^*\big] & \geq 0 \text{ if } M\ge\left(1+\epsilon\right)n,\text{ } \epsilon>0.
\end{align*}
-- If $\gamma > \dfrac{1}{\beta-1}$,
\begin{align*}
	\EE\big[\widetilde{R}_{z_{>1}}^*\big] &= \Omega\left(m^{1-\mu(\beta-1)}\right)\text{if } M=\oo\left(m^{\frac{1}{\beta-1}}\right), \\
	\EE\big[\widetilde{R}_{z_{>1}}^*\big] &\geq 0 \text{ if } M=\Omega\left(m^{\frac{1}{\beta-1}}\right).
\end{align*}
\color{black}
\end{theorem}

\begin{remark}
	Note that this result only depends on the cumulative cache memory and is valid for all storage profiles with the same amount of cumulative cache memory. This result is a generalization of a result in \cite{moharir2016content} which holds only if the number of files scales linearly with the number of caches, i.e., $\gamma=1$. 
	
\end{remark}

\subsection{Knapasack Storage + Match Least Popular Policy}

Next, we analyze the performance of a policy called Knapsack Storage + Match Least Popular policy (KS+MLP), proposed in \cite{moharir2016content}.  In \cite{moharir2016content}, it was shown that the KS+MLP policy is orderwise optimal for the homogeneous setting if the number of caches scales linearly with the number of files. \color{black} We first make suitable modifications to the policy to incorporate heterogeneity in memory across caches and analyze its performance for more general storage profiles. We describe the modified KS+MLP policy in detail for the sake of completeness.

The KS+MLP policy comprises of two phases: the \textit{placement phase} and the \textit{delivery phase}.  
\subsubsection{Placement Phase}
In the \textit{placement phase}, the goal is to determine what to store on each cache. This task is completed in two steps. 

%

\noindent \textbf{Knapsack Storage: Part 1} -- In this part, we decide how many caches store copies of each content by solving a Fractional Knapsack problem. In the Fractional Knapsack problem, each object has two attributes, namely, a weight and a value, and the knapsack has a finite weight capacity. The goal is to determine which objects should be added to the knapsack to maximize their cumulative value while the weight constraint of the knapsack is not violated. In the KS+MLP policy, each file corresponds to an object. The weight of an object/file corresponds to the number of caches on which it will be replicated if selected. The weights are chosen such that with high probability, all requests for that file can be served using the caches. More specifically, if file popularity follows the Zipf distribution with parameter $\beta > 1$, the weight of File $i$, denoted by $w_i$ is assigned the following values. 
\begin{align}\label{eq:wi}
w_{i} = \begin{cases}
{m}, & \text{if } i = 1 \\
\big\lceil \big(1 + \frac{p_1}{2}\big) \widetilde{m}p_i \big\rceil, & \text{if } 1 < i \leq n_1 , \\
\big\lceil  {4 p_1(\log m)^2} \big\rceil, & \text{if } n_1 < i \leq n_2, \\
\big\lceil{\frac{1}{\delta}+1}\big\rceil, & \text{if } n_2 < i \leq n, 
\end{cases} 
\end{align}
where $n_1=\frac{(\widetilde{m}p_1)^\frac{1}{\beta}}{(\log m)^\frac{2}{\beta}}$, and $n_2=m^{\frac{1+\delta}{\beta}}$ for some $\delta>0$.
The value of File $i$ is the probability that it is requested at least once in a batch of requests. The weight capacity of the cache system is equal to the cumulative cache memory. Using these parameter values, we solve the following Fractional Knapsack problem: 
	\begin{eqnarray*}
	\max &&  \displaystyle \sum_{i=1}^n x_i (1-(1-p_i)^{\widetilde{m}}) \\
	\text{s.t. }  && \displaystyle \sum_{i=1}^n x_i w_i  \leq M, \\
	&& 0 \leq x_i \leq 1, \text{ } \forall i.
\end{eqnarray*}

 If the solution to the above Knapsack problem gives $x_i=1$, we  store $w_i$ copies of File $i$ else we don't store File $i$. 
 
%
%
%

\textbf{Knapsack Storage: Part 2} -- The previous step determines how many copies of each file will be stored on the caches. The next task is to store the copies of files on caches. File copies are sorted in increasing order of the corresponding file index. For example, consider a system consisting of five caches with $k_1=3, k_2=2, k_3=2, k_4=1, k_5=1$ units of memory. Say the solution for Knapsack Storage: Part 1  gives $x_1 = x_2 = x_3 = x_4 = x_5 = 1$ and 0 otherwise, and $w_1 = 4$, $w_2=2$, $w_3=w_4= w_5= 1$. The sorted list of file copies is illustrated in Figure \ref{fig:knapsack_storage_Part 2_example}. Recall that caches are indexed in decreasing order of memory. The sorted list of file copies is stored on the caches in a round robin manner, i.e., the next file copy is placed on the next cache which has a memory slot available, see Figure \ref{fig:knapsack_storage_Part 2_example} for an illustration. 
%
%
\begin{figure}[h]
\begin{center}
	\textbf{Sorted S:} \\
	\begin{tabular}{| c | c | c | c | c | c | c | c | c |}
		\hline
		$1$ & $1$ & $1$ & $1$ &  $2$ & $2$ & $3$  & $4$  & $5$  \\
		\hline
	\end{tabular} 
\vspace{0.1in}

	\textbf{Storage:}

	\includegraphics[height=2cm]{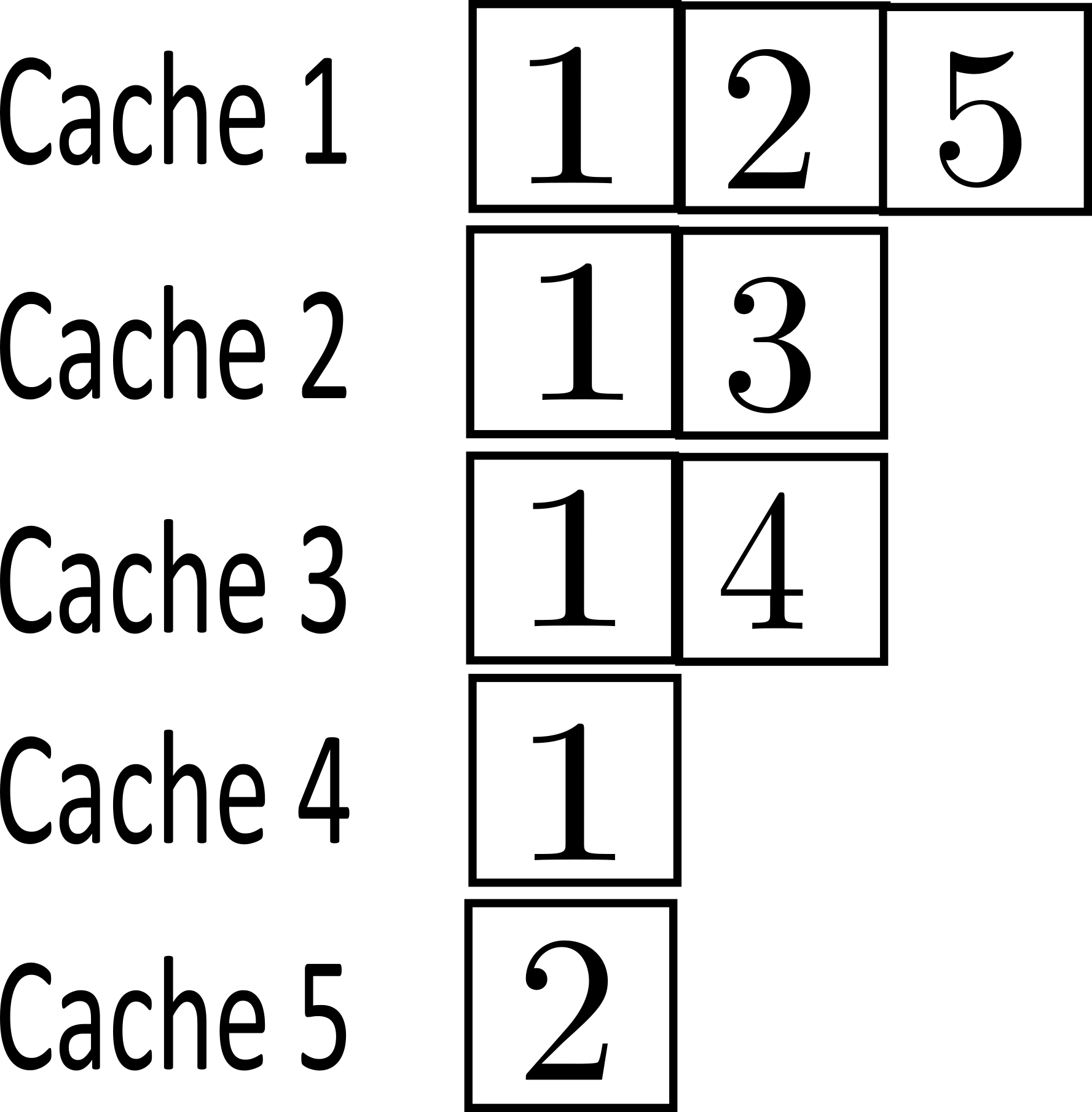}
%
%
%
%
%
%
%
	\caption{Illustration of Knapsack Storage: Part 2 for a system with five caches.}\label{fig:knapsack_storage_Part 2_example}
\end{center}
\end{figure}

\subsubsection{Delivery Phase}
In the \textit{delivery phase}, requests are allocated to caches for service using the Match Least Popular policy (MLP), such that each cache is matched to at most one request. All the matched requests are served by the corresponding caches and all the unmatched requests are assigned to the central server. {As the name suggests,} the Match Least Popular policy matches requests for unpopular files before matching requests for popular files to caches. Refer to Figure \ref{fig:match_least_popular} for a formal definition. 

\begin{figure}[h]
	\hrule
	\vspace{0.1in}
	\begin{algorithmic}[1]
		\STATE initialize $i = n$, set of idle caches $= \{1,2,..., m\}$. 
		\IF {the number of requests for File $i$  is more than the number of idle caches storing File $i$,}
		\STATE goto Step 8.
		\ELSE 
		\STATE match requests for File $i$ to idle caches storing File $i$, chosen uniformly at random. 
		\STATE update the set of idle caches. 
		\ENDIF
		\STATE $i = i-1$, goto Step 2.
	\end{algorithmic}
	\vspace{0.1in}
	\hrule
	\caption{Match Least Popular -- \sl Matches requests to caches.}
	\label{fig:match_least_popular}
\end{figure}
\color{black}

The next theorem evaluates the performance of the KS+MLP policy for a particular sub-class of heterogeneous distributed cache systems. 
\begin{theorem} \label{thm:ach_het}
	Consider a distributed cache system satisfying Assumption \ref{ass:system} where file popularity follows the Zipf distribution with parameter $\beta> 1$ and the top (largest) $\Omega (m^{2-\beta+\delta})$ caches, for any $\delta>0$ have the same storage size. We have no restrictions on the storage sizes of the other (smaller) caches. Let $\EE\big[\widetilde{R}_{z_{>1}}^{\text{KS}}\big]$ denote the expected transmission rate of the KS+MLP policy  for this system. \\
	-- If $\gamma \leq \dfrac{1}{\beta-1}$
	\begin{eqnarray*}
		\mathbb{E}\big[\widetilde{R}_{z_{>1}}^{\text{KS}}\big]
		= \begin{cases}
			\OO\left(m^{1-\mu(\beta-1)}\right)&\text{if } M\leq(1-\epsilon)n,\text{ } 0<\epsilon<1, \\
			\OO\left(m^{\frac{2-\mu\beta}{\beta}}\right)&\text{if } M=n\\
			\OO\left(1\right)&\text{if } M\ge\left(1+\epsilon\right)n,\text{ } \epsilon>0.
		\end{cases}
	\end{eqnarray*}
	-- If $\gamma > \dfrac{1}{\beta-1}$
	\begin{eqnarray*}
		\mathbb{E}\big[\widetilde{R}_{z_{>1}}^{\text{KS}}\big]
		= \begin{cases}
			\OO\left(m^{1-\mu(\beta-1)}\right)&\text{if } M=\oo\left(m^{\frac{1}{\beta-1}}\right), \\
			\OO\left(1\right)&\text{if } M=\Omega\left(m^{\frac{1}{\beta-1}}\right).
		\end{cases}
	\end{eqnarray*}
	 
\end{theorem}

\begin{remark}
	The key takeaways from this result are:
	\begin{itemize}
		\item[--] If the top $\Omega (m^{2-\beta+\delta})$ caches, for any $\delta>0$ have the same memory size, then KS+MLP results match orderwise with the lower bounds in Theorem \ref{thm:conv_het}. Hence, in this case, the KS+MLP policy is orderwise optimal.
		\item[--] Homogeneous systems have all caches with equal memory, and thus Theorem \ref{thm:ach_het} also holds for homogeneous systems. 
	\end{itemize}
\end{remark}

Combining Theorems \ref{thm:conv_het} and \ref{thm:ach_het} we have the following result. 

\begin{corollary}
	\label{cor:zipf_morethan1}
	Consider a distributed cache system satisfying Assumption \ref{ass:system} where file popularity follows the Zipf distribution with parameter $\beta> 1$. 
	If the top $\Omega (m^{2-\beta+\delta})$ caches, for any $\delta>0$ have the same memory size, the performances of the optimal schemes for heterogeneous and homogeneous systems are orderwise equal.
\end{corollary}

Compare the above result with Corollary~\ref{Cor:HetHom} for $\beta < 1$, which considers a heterogeneous cache system that divides a cumulative cache memory of $M = \Theta(n\cdot m^{\delta}), \delta > 0$, amongst two classes of caches: `\textit{rich}' caches with larger storage size and `\textit{poor}' caches with smaller storage. Corollary~\ref{Cor:HetHom} shows that even if only a constant fraction of the caches are restricted to be poor, it can cause significant disparity between  the performances of heterogeneous and homogeneous systems. On the other hand, in the same setting for $\beta > 1$, Corollary~\ref{cor:zipf_morethan1} shows that even if as many as $m -  \Omega(m^{2-\beta+\delta})$ caches are restricted to be poor with only one unit of memory, the performance of the system will be orderwise the same as the homogeneous system. This suggests that caching systems are more tolerant to heterogeneity in  storage  under Zipf distributions with parameter $\beta> 1$  than under Zipf distributions with parameter $\beta < 1$. 


Intuitively, this difference can be explained as follows. When content popularity is lopsided  ($\beta > 1$), under the KS+MLP policy, caches with limited storage are used to serve requests for popular contents and the caches with large storage which store a mixture of some popular and a large number of unpopular contents typically are allocated to serve requests for unpopular contents. This ensures that the low storage caches are also utilized, unlike the case when content popularity is comparable across files. Since most caches are utilized, the number of requests served by the central server is small. As a result, the effect of storage heterogeneity is lower for lopsided content popularity distributions as compared to distributions where it is comparable across files. 

\color{black}
%

Corollary \ref{cor:zipf_morethan1} describes a sufficient condition under which the performances of the homogeneous and heterogeneous systems remain comparable. Our next result characterizes a degree of heterogeneity sufficient to ensure that the performance of the heterogeneous system is orderwise inferior to that of a homogeneous system with the same amount of cumulative cache memory. 

%
%
%

\begin{theorem}\label{thm:oneuserpercache}
	Consider a distributed cache system satisfying Assumption \ref{ass:system} where file popularity follows the Zipf distribution with parameter $\beta> 1$. 
	Let $\widetilde{R}_{z_{>1}}^*$ denote the optimal transmission rate for any uncoded storage/service policy.	If $\exists$ a subset $\mathcal{S}$ of caches with cumulative memory $|M_{\mathcal{S}}|$ such that
	\begin{itemize}
		\item [--]{ $|\mathcal{S}|\geq m-m^{1-\mu(\beta-1)-\delta}, \  \text{for any }\delta>0$ and} 
		\item [--]	{$|M_{\mathcal{S}}|\leq (1-\epsilon)n,  \  \text{for any }\epsilon>0$,  $$\text{then, } \mathbb{E}[\widetilde{R}_{z_{>1}}^*]\geq \Omega\big(m^{1-\mu(\beta-1)}\big).$$}
	\end{itemize} 
\end{theorem}

We thus conclude that if there is a large enough set of caches (with cardinality $m-m^{1-\mu(\beta-1)-\delta}, \  \text{for any }\delta>0$) with cumulative storage less than a constant fraction of the catalog size, the expected transmission rate can't be made arbitrarily small, irrespective of the total cache memory in the system. 

\subsubsection*{Example} Consider a heterogeneous distributed cache system with $m$ caches and $n=cm$ $(c>1)$ files with content popularity following the Zipf distribution with $\beta>1$. We have two classes of caches: `rich' and `poor'. Let the total cumulative memory in the system be $M = (1+\epsilon)n = (1+ \epsilon)cm$ for some $\epsilon>0$. Let $m_1$ denote the number of rich caches, each of which has $k \gg 1$ units of memory. The remaining $m - m_1$ poor caches each have $1$ unit of memory, see Figure~\ref{fig:het3} for an illustration. Thus, we have $m-m_1+m_1k = M = (1+\epsilon)n$. For {some} small $\delta>0$, Figure~\ref{fig:het3} depicts two systems with the same total cumulative memory, in which the number of rich caches is $m_1=m^{2-\beta-\delta}$ and $m_1=m^{2-\beta+\delta}$ respectively. For the former system which has fewer number of rich caches, the expected rate grows as $\Omega(m^{2-\beta})$ from Theorem~\ref{thm:oneuserpercache}. On the other hand, for the latter system which has more rich caches,   Corollary~\ref{cor:zipf_morethan1} shows that the KS+MLP policy achieves $\oo(1)$ rate. So for {some} small $\delta$, modifying the storage profile to change the number of rich caches from $m^{2-\beta-\delta}$ to $m^{2-\beta+\delta}$ can have a dramatic impact on the server transmission rate. 
\color{black}


\begin{figure}[h]
	\centering
	\includegraphics[scale=.45]{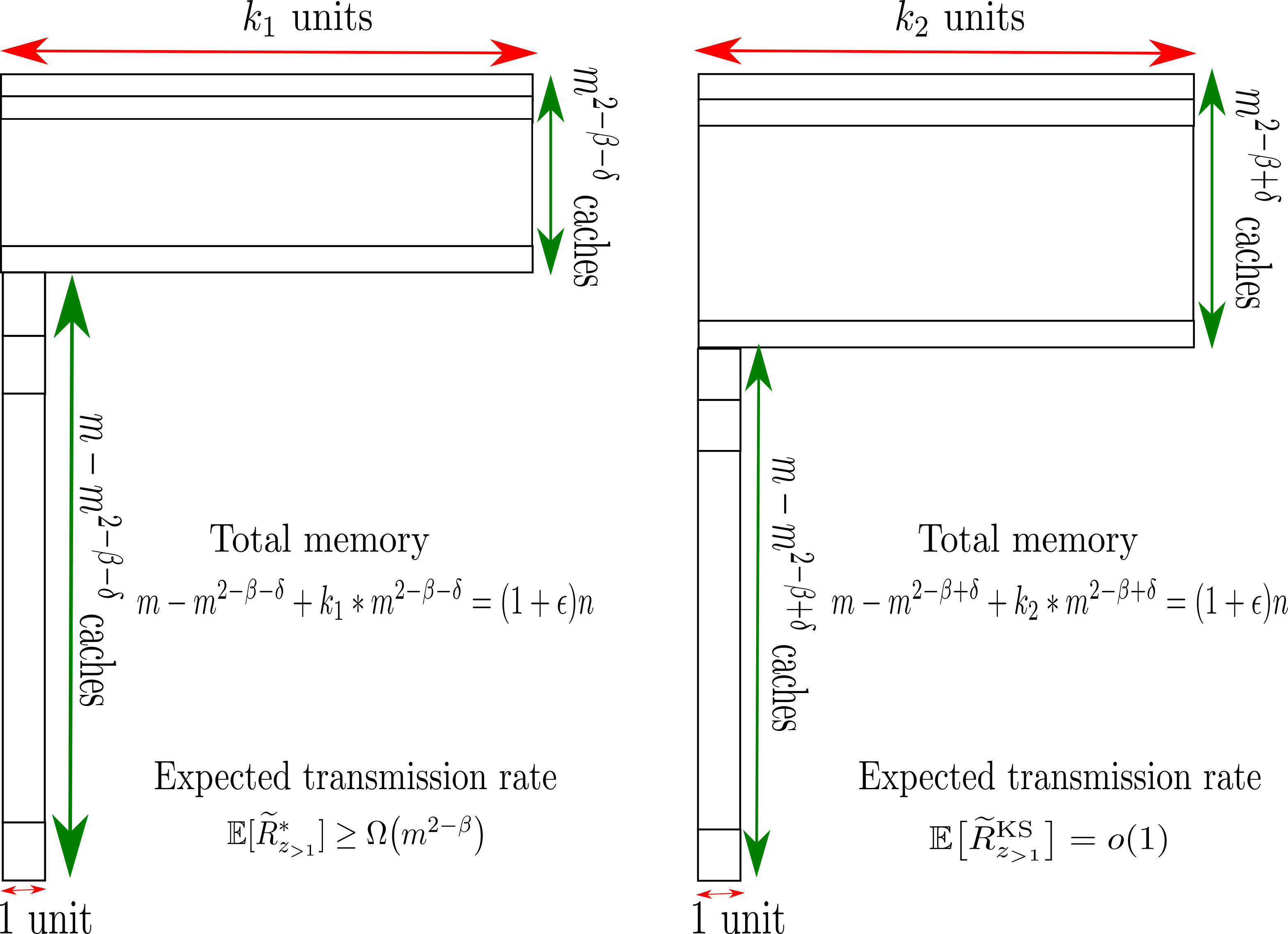}
	\caption{Impact of storage heterogeneity on server transmission rates}
	\label{fig:het3}
\end{figure}

%% file: simulations.tex
\section{Simulation Results}
\label{sec:simulations}
In Section \ref{sec:results}, we presented asymptotic results as the system size $m$ grows, which compare the effects of storage heterogeneity on the server transmission rate as a function of $\beta$ (or as a function of the popularity profile). In this section, we simulate finite size cache systems and empirically validate some of our theoretical findings.  

First, we consider a system which consists of $m$ caches with total memory $M$ units, $n=m$ files with popularity following the Zipf distribution with $\beta=0.3$, and $\widetilde{m}=0.97m$ requests. Similar to the example in the previous section, we consider two classes of caches: `rich' and `poor',  i.e., out of the $m$ caches, $m_1$ caches (rich caches) each have $k$ units of memory and the remaining $m-m_1$ caches (poor caches) each  have only $1$ unit of memory. As the value of $m_1$ decreases, the  memory is concentrated among fewer caches.  For this system, we simulate the PPMM policy as described in Section~\ref{sec:lessthanone} and consider the server transmission rate, averaged over 100 experiments. 

\begin{figure}[h]
		\centering
		\includegraphics[width = 0.47\textwidth]{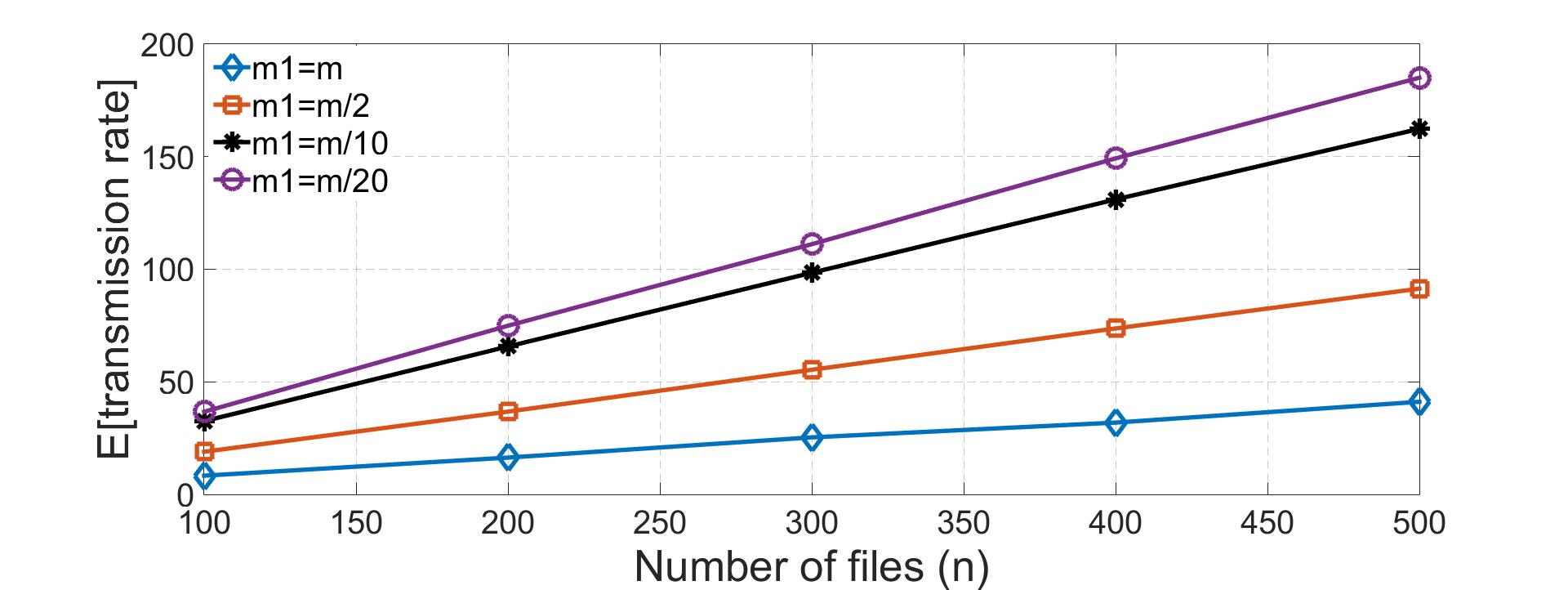}
		\caption{Plot of the average transmission rate vs the number of files ($n$) for PPMM policy with $m_1=\{m, \frac{m}{2},\frac{m}{10},\frac{m}{20}\}$, for a system where the number of caches ($m$) = $n$, the number of requests ($\widetilde{m}$) = $0.97n$, the Zipf parameter ($\beta$) = $0.3$, and  the total memory ($M$) = $3n$.}\label{fig:1} 
\end{figure}

In Figure~\ref{fig:1}, we fix the total memory to $M = m_1k+m-m_1=3m$ units and plot the average transmission rate as a function of number of files $n$ for various values of $m_1$. As expected, \emph{(i)} the transmission rate increases  with $n$, and \emph{(ii)} for any fixed value of $n$, the transmission rate increases drastically as the number of rich caches $m_1$ decreases. As  our result in Corollary~\ref{Cor:HetHom} suggests, there is significant difference between the homogeneous and heterogeneous cases.

\begin{figure}[h]
	\centering
	\includegraphics[width = 0.47\textwidth]{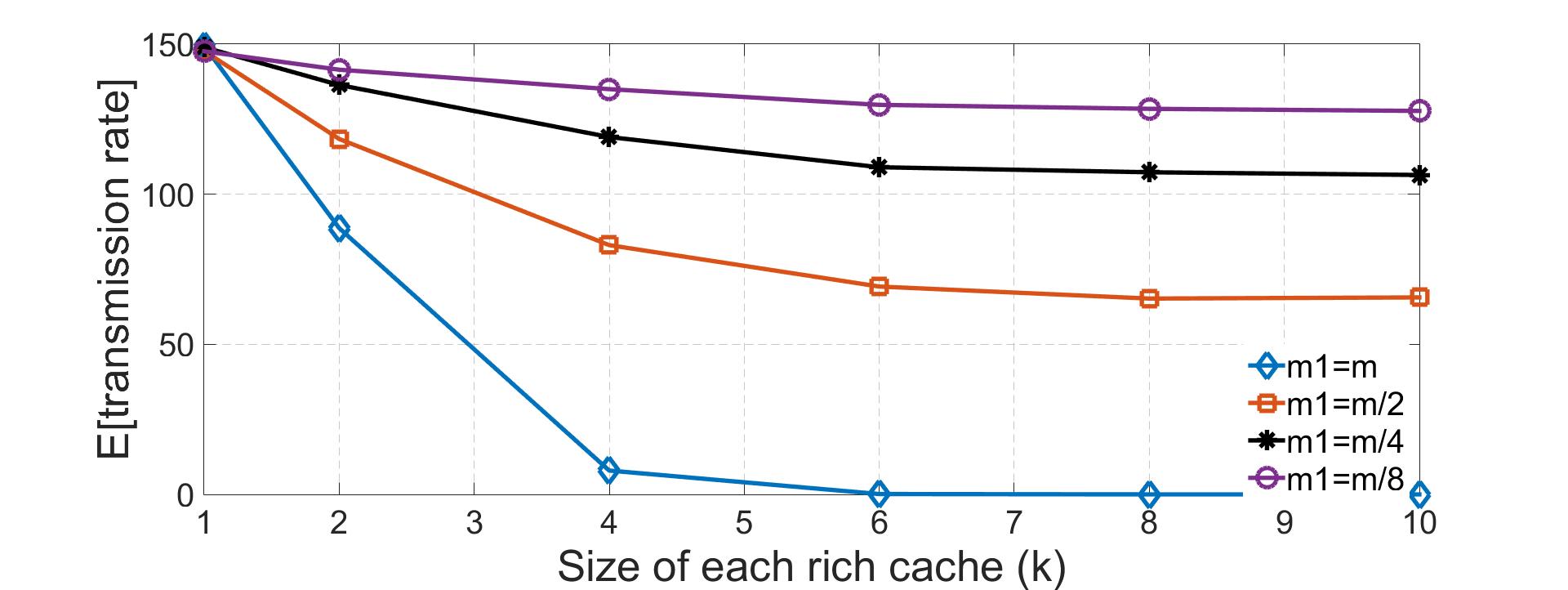}
	\caption{Plot of the mean transmission rate vs cache size (k) of  each of the $m_1$ rich caches for PPMM policy with $m_1=\{m, \frac{m}{2},\frac{m}{4},\frac{m}{8}\}$, for a system where the number of caches ($m$) = the number of files ($n$) = $400$, the number of requests ($\widetilde{m}$) = $0.97n$, and the Zipf parameter ($\beta$) = $0.3$.}\label{fig:2} 
\end{figure}

In Figure \ref{fig:2}, we fix $m$ $=$ $n$ $=$ $400$ and plot the average server transmission rate as a function of the cache size $k$ of each of the $m_1$ rich caches, for various values of $m_1$. As we increase $k$, we expect the transmission rate to decrease initially until all the rich caches serve one request each, and remain constant thereafter since the storage capacity of the poor caches is fixed throughout to $1$ unit. As expected,  \emph{(i)} for the homogeneous case, the average transmission rate decreases exponentially with $k$ until it reaches 0, and  \emph{(ii)} for the heterogeneous case, the average transmission rate decreases initially and remains constant after a certain $k$, depending upon the heterogeneity level ($m_1$).

Next, we consider a system which consists of $m$ caches with total memory $M$ units, $n=5m$ files with popularity following the Zipf distribution with $\beta=1.2$, and $\widetilde{m}=0.97m$ requests. As before, we consider $m_1$ rich caches and $m - m_1$ poor caches. For this system, we simulate the KS+MLP policy as described in Section~\ref{section:morethanone} and consider the server transmission rate, averaged over 100 iterations. 

\begin{figure}[h]
	\centering
	\includegraphics[width = 0.47\textwidth]{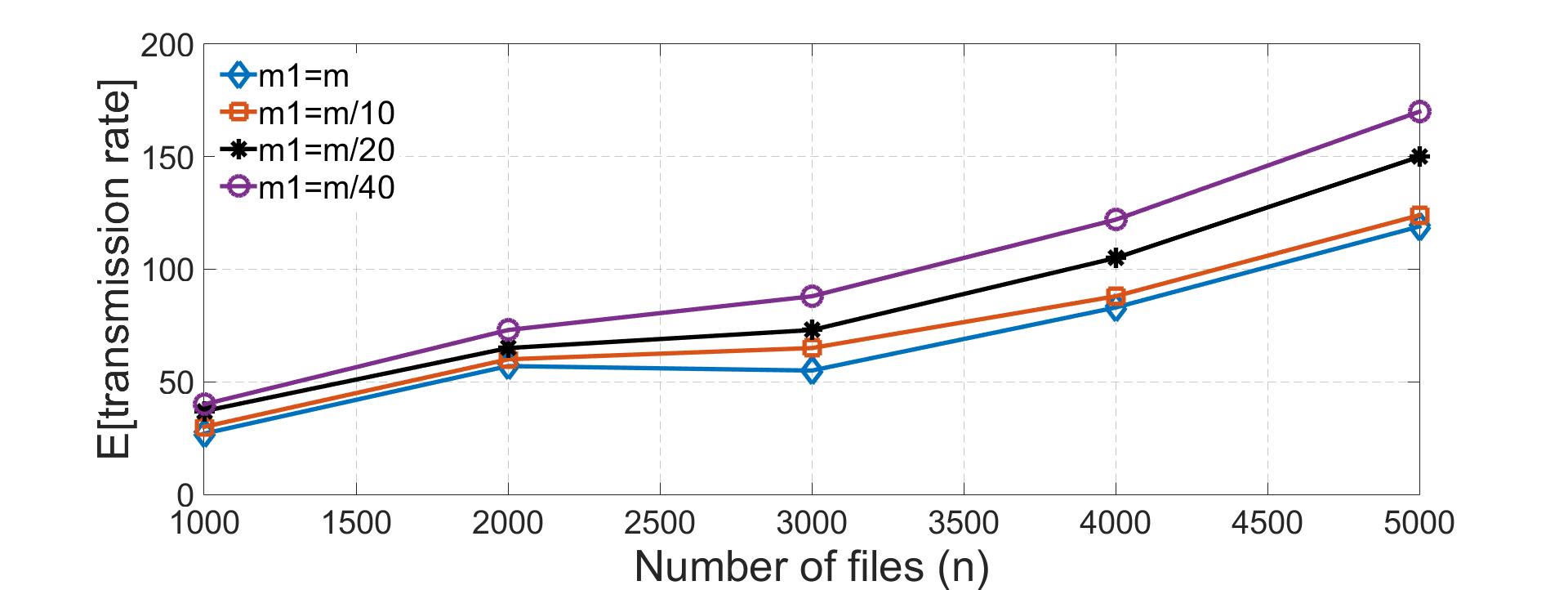}
	\caption{Plot of the average transmission rate vs the number of files ($n$) for KS+MLP policy with $m_1=\{m, \frac{m}{10},\frac{m}{20},\frac{m}{40}\}$, for a system where the number of caches ($m$) = $\frac{n}{5}$, the number of requests ($\widetilde{m}$) = $0.97n$, the Zipf parameter ($\beta$) = $1.2$, and  the total memory ($M$) = $3n$.}\label{fig:3} 
\end{figure}

In Figure \ref{fig:3}, we fix the total memory to $M =m_1k+m-m_1=3m$ units and plot the average transmission rate as a function of number of files $n$ for various values of $m_1$. As expected, \emph{(i)} the transmission rate increases with $n$, and \emph{(ii)} for any fixed value of $n$, unlike the $\beta=0.3$ case (plotted in Figure~\ref{fig:1}), the change in transmission rate  for different values of $m_1$ is small. This is in line with our result in Corollary~\ref{cor:zipf_morethan1}, which suggests that the performances of the  homogeneous and heterogeneous systems are similar.

%% file: proofs2.tex
\section{proofs} \label{sec:proofs}
In this section, we prove  the results stated in Section~\ref{sec:results}{\footnote{Homogeneous basic versions of proofs are given in \cite{moharir2016content, reddy2017resource}.}}. We are interested in order wise results. {In the rest of this section, {we will use $c_{i}$, where $i \in \mathbb{N}$, to represent positive constants.}

\subsection{Proof of Theorem \ref{thm:zipf01}}
We analyze the performance of PPMM policy discussed in Section \ref{sec:results} for $\beta \in [0,1)$ using ideas from the proof of Proposition 1 in \cite{LLM12}, which looks at the  similar setting where the request arrival process is Poisson and $\gamma=1$.\\
	
	\begin{proof} (Theorem \ref{thm:zipf01})

	\textbf{Case 1:} $M<(1-\epsilon)n$, $\epsilon>0$

	Trivial.

	\textbf{Case 2:} $M\geq n$
	
		Recall that each file is stored proportional to its popularity in PPMM policy. $i.e.,$ File $i$ is stored on $d_i=Mp_i$ caches. Let $b_i$ denotes the number of requests for File $i$ in a batch. $b_i$ is Bin($\widetilde{m},p_i$), where, $p_i=\frac{p_1}{i^{\beta}}$.  For $b_i$ requests of File $i$, we split each request into $d_i$ sub-requests, each of size $\frac{1}{d_i}$ units. 
	Let $\partial s$ denote the set of files stored on Cache $s$. For each $i \in \partial s$, we associate $b_i$ sub-requests of File $i$ to Cache $s$. This allocation leads to a fractional matching where the  total data served by each cache is less than $1$ unit if $\forall s\in \{ 1,2,...,m\}$, 
	\begin{align}
	\sum_{i\in \partial s}\frac{b_i}{d_i}\leq 1.\nonumber
	\end{align}
	
	By the total unimodularity of adjacency matrix, the existence of a fractional matching implies the existence of an integral matching \cite{LLM12}.

	For the Zipf distribution with parameter $\beta\in[0,1)$, we have that,
	\remove{\color{green}\begin{align}\label{eq:3.1}
	p_n=\frac{n^{-\beta}}{1+2^{-\beta}+...+n^{-\beta}}.
	\end{align}
	Note that,
	\begin{align} \label{eq:3.2}
	1+2^{-\beta}+...+n^{-\beta}\leq \int_0^{n}i^{-\beta}di 
	=\frac{n^{-\beta+1}}{1-\beta}.
	\end{align}
	From \eqref{eq:3.1} and \eqref{eq:3.2}, we get that,}
	\begin{align}
	p_n \geq 
	\frac{{1-\beta}}{n}. \nonumber
	\end{align}
	
	Let $p^*=\frac{{1-\beta}}{n}$. Hence, $p_i\geq p^*\hspace{0.5in}\forall i.$
	
	\begin{align}
	\mathbb{P}\bigg(\sum_{i\in \partial s}\frac{b_i}{d_i}>1\bigg)&\leq \inf_{s>0}\text{ }\frac{\mathbb{E}\bigg[e^{s\sum_{i\in \partial s}\frac{b_i}{d_i}}\bigg]}{e^s} \nonumber\\
	&\leq \inf_{s>0} \text{ }e^{-s} \prod_{i\in \partial s}\mathbb{E}\Big[e^{\frac{sb_i}{d_i}}\Big]\nonumber\\
	&\text{(due to negative associativity \cite{dubhashi1996balls})}\nonumber\\
	&= \inf_{s>0}\text{ }e^{-s} \prod_{i\in \partial s} e^{\widetilde{m}\ln\big(p_ie^\frac{s}{d_i}+1-p_i\big)}\nonumber\\
	&\text{Since, above function is decreasing in $p_i$}\nonumber\\
	&\leq \inf_{s>0}\text{ } e^{-s} \prod_{i\in \partial s}
	e^{\widetilde{m}\ln\big(p^*e^\frac{s}{Mp^*}+1-p^*\big)}\nonumber\\
	&=  \inf_{s>0}\text{ }e^{-s+\widetilde{m}\frac{M}{m}\ln\big(p^*e^\frac{s}{amkp^*}+1-p^*\big)}\nonumber\\
	&=e^{-\frac{M(1-\beta)\rho}{n}h(\frac{1}{\rho})}\nonumber
	\end{align}
	where $h(x)=x\ln x -x+1$ is the Cramer transform of a unit Poisson random variable.
	Taking the Union bound over all $m$ caches,
	\begin{equation*}
	\mathbb{P}\big(\text{matching exists}\big) \geq 1-me^{-\frac{M(1-\beta)\rho}{n}h(\frac{1}{\rho})}.
	\end{equation*}
	Then, 
	\begin{align*}
	\mathbb{E}[R_{z_{[0,1)}}^{\text{PPMM}}] &\leq 0 \times \mathbb{P}\big(\text{matching exists}\big)\\
	&\hspace{0.5in}+ \widetilde{m} \times (1- \mathbb{P}\big(\text{matching exists}\big)\big) \\
	&\leq \widetilde{m}me^{-\frac{M(1-\beta)\rho}{n}h(\frac{1}{\rho})},
	\end{align*}
	and the result follows.	
	\end{proof}\\

\subsection{Proof of Theorem \ref{thm:het01}}

We use the following lemmas to prove Theorem \ref{thm:het01}.\\

\begin{lemma}
	\label{lemma:chernoff}
	For a Binomial random variable $X$ with mean $\mu$, by the Chernoff bound, $\forall$ $\delta\geq0$,
	\begin{eqnarray*}
		&& \PP(X \geq (1+\delta) \mu) \leq \Bigg(\frac{e^\delta}{(1+\delta)^{(1+\delta)}}\Bigg)^\mu, \\
		&& \PP(X \leq (1-\delta) \mu) \leq e^{-\delta^2 \mu/2}.
	\end{eqnarray*}
\end{lemma}

\begin{proof}
	{Follows from the properties of the Binomial distribution.}
\end{proof}\\

{\begin{corollary}\label{cor:chernoff}
	For a Binomial random variable $X$ with mean $\mu$, for all $0\leq \delta \leq 1$,
	$$\PP(X \leq (1+\delta) \mu) \leq e^{-\delta^2 \mu/3}.$$
\end{corollary}

\begin{lemma} \label{lemma:max_requests}
	For a content delivery system satisfying Assumption \ref{ass:system} with file popularity follows the Zipf distribution with Zipf parameter $\beta\in [0,1)$, let $d_i$ represents the number of requests for File $i$ in a batch.  Then, for any $\delta>0$, $$\PP\big(d_i \leq 2m^{\max\{0,1-\mu(1-\beta)\}+\delta}\big)=\OO\Big(e^{-m^{\max\{0,1-\mu(1-\beta)\}+\delta}}\Big).$$
	
\end{lemma}

\begin{proof}
	The popularity of File 1 is $p_1=\frac{1}{\sum_{i=1}^{n}i^{-\beta}}\leq \frac{1}{n^{1-\beta}}$ for large $n$. Under Assumption \ref{ass:system}, the number of requests for File 1 is Bin($\widetilde{m},p_1$) and the expected number of requests is $\leq m^{1-\gamma(1-\beta)}$. Consider a new binomial random variable $X$ with mean $m^{\max\{0,1-\gamma(1-\beta)\}+\delta}$. By Corollary \ref{cor:chernoff},
	\begin{align*}
	 &\PP\big(d_i \leq 2m^{\max\{0,1-\gamma(1-\beta)\}+\delta}\big)\\
	 &\hspace{1.4in}\leq \PP\big(X \leq 2m^{\max\{0,1-\gamma(1-\beta)\}+\delta}\big)\\ &\hspace{1.4in}=\OO\Big(e^{-m^{\max\{0,1-\gamma(1-\beta)\}+\delta}}\Big).
	 \end{align*}
\end{proof}}\\

We now prove the Theorem \ref{thm:het01}, which evaluates the performance of heterogeneous system with $\beta \in [0,1)$.\\

\begin{proof}(Theorem \ref{thm:het01} -- \emph{Lower bound on transmission rate})
	
	If  a cache stores $k$ units of data, the probability of the cache is being idle is $\geq \big(1-(\frac{k}{n})^{1-\beta}\big)^{\widetilde{m}}.$ If $(\frac{k}{n})^{1-\beta}=\frac{1}{c_1m} $ i.e., $k=\Theta(m^{\gamma-\frac{1}{1-\beta}})$, and $c_2m$ caches have memory less than $k$ units, then the expected number of idle caches is $\geq c_2me^{-\frac{\rho}{c_1}}$. Hence, the number of unserved requests is $\geq \widetilde{m}-m+c_2me^{-\frac{\rho}{c_1}}$. If $c_2>\frac{1-\rho}{e^{-\frac{\rho}{c_1}}}$, then the number of unserved requests is $\Theta(m).$
	From Lemma \ref{lemma:max_requests}, no file is requested more than $m^{\max\{0,1-\gamma(1-\beta)\}+\delta}$ times for any  $\delta>0$. Hence, the expected transmission rate between the central server and the root-node is $\Omega\big(\frac{m}{m^{\max\{0,1-\gamma(1-\beta)\}+\delta}}\big)=\omega(1)$.
\end{proof}\\

\begin{proof} (Theorem \ref{thm:het01} -- \emph{Performance of PPMM})
	
	Consider a new system (System B) by ignoring low memory ($k_i<m^{\gamma-\delta}$) caches, i.e., System B contains $m'=(\rho+c_3)m$ caches,  $n={(\frac{m'}{\rho+c_3})}^\gamma$ files, total memory is $M=m'*m^{\gamma-\delta}$, and receives $\widetilde{m'}=\frac{\rho}{\rho+c_3}m'$ requests. Let $\EE[R_B]$ is expected transmission rat of System B. System B satisfies the conditions of Corollary \ref{cor:zipf_0to1}. Therefore,  $$\EE\big[\widetilde{R}^{\text{PPMM}}_{z_{[0,1)}}\big]\leq \EE[R_B]=\oo(1).$$ 	
\end{proof}\\

\subsection{Proof of Theorem \ref{thm:conv_het}}
We use the following lemma and proposition to prove Theorem \ref{thm:conv_het} and Theorem \ref{thm:ach_het}.\\


\begin{lemma}
	\label{lemma:Zipf_d}
	Let a content delivery system contains $n=m^{\gamma}$ files, with file popularity follows the Zipf distribution with Zipf parameter $\beta>1$. In a given time-slot, let $d_i$ represents the number of requests for File $i$. Let $E_1$ be the event that:
	\begin{enumerate}
		{\item[(a)] $ d_i \geq 1$ for $i=\OO\big(m^{\frac{1}{\beta}-\delta}\big)$, where $\epsilon>0$ is arbitrarily small constant,}
		\item[(b)] $d_i \leq 2p_1 (\log m)^2$ for $n_1 < i \leq n_2$,
		\item[(c)] $d_i \leq \Big(1+\dfrac{p_1}{4}\Big) mp_i$ for $1 \leq i \leq n_1$,
	\end{enumerate}
	where $n_1=\frac{(\widetilde{m}p_1)^\frac{1}{\beta}}{(\log m)^\frac{2}{\beta}}$, and $n_2=m^{\frac{1+\delta}{\beta}}$ for {{red}some} $\delta>0$.
	Then we have that, 
	$$\PP(E_1) = 1-\OO\Big(n e^{-(\log m)^{2}}\Big).$$
\end{lemma}
\begin{proof} 	
	
	Follows from the Chernoff bound. \cite{moharir2016content} proves similar lemma (Lemma 2).
%
%
\end{proof}\\

\color{black}

\begin{proposition}
	\label{prop:converse}
	Consider a distributed cache system satisfying Assumption \ref{ass:system} where file popularity follows the Zipf distribution with parameter $\beta> 1$. Let File $i$ size is $F_i$ bits, and $R^*$ denote the minimum transmission rate under the constraint that no stored bit can serve more than one user. Then, we have,	
	\begin{eqnarray*}
		\EE\left[R^*\right] &\geq& 
		\displaystyle \sum_{i=1}^n \displaystyle \sum_{u=1}^{F_i}\left(1-(1-p_i)^{\widetilde{m}}\right) - \text{O}^*, \\
		\text{where, O}^* &=& \max \displaystyle \sum_{i=1}^n \displaystyle \sum_{u=1}^{F_i} x_{i,u} \left(1-(1-p_i)^{\widetilde{m}}\right) \\
		&& \text{s.t. }  \displaystyle \sum_{i=1}^n  \displaystyle \sum_{u=1}^{F_i} x_{i,u}\max\{\widetilde{m}p_i,1 \}  \leq M, \\
		&& 0 \leq x_{i,u} \leq 1, \text{ } \forall i,u.
	\end{eqnarray*}
\end{proposition}

\begin{proof}
	\cite{moharir2016content} proves similar theorem (Theorem 1).
\end{proof}\\

\begin{remark}
Proposition \ref{prop:converse} gives a lower bound to our system. The lower bound depends on the solution to the fractional knapsack problem ($\text{O}^*$). In particular, if the content popularity follows the Zipf distribution with parameter $\beta$ (i.e., File $i$ is requested with popularity proportional to the $i^{-\beta}$), then the optimal solution has following structure: $\exists$ $i_{\min}$, $i_{\max}$ such that 
\begin{eqnarray*}
	x_{i,u} = \begin{cases}
		1, & \text{  if } i_{\text{min}} < {{i}} < i_{\text{max}}, \\
		f_1, & \text{if } i =i_{\text{min}} , \\
		f_2, & \text{if } i =i_{\text{max}} , \\
		0, & \text{otherwise},
	\end{cases}
\end{eqnarray*}
where, $0\leq f_1, f_2 \leq 1$, $$\displaystyle \sum_{i=1}^n  \displaystyle \sum_{u=1}^{b_i} x_{i,u}\max\{\widetilde{m}p_i,1 \}  \leq M.$$
It follows that
$$ \EE[R^*] =\Omega \left(\displaystyle \sum_{i=1}^n \displaystyle \sum_{u=1}^{b_i}(1-x_{i,u})(1-(1-p_i)^{\widetilde{m}})\right).$$
\end{remark}

\begin{proof} (Theorem \ref{thm:conv_het})
	
\textbf{Case 1:} $M\leq(1-\epsilon)n,\text{ } 0<\epsilon<1$

Consider a new system with one cache of size $M$ units which can serve all the requests for the stored contents. It is clear that a lower bound on the transmission rate in the new system is also a lower bound on the transmission rate of the original system.

In the new system, we can store at most $M$ files. Therefore, all requests for the $n-M\geq\epsilon n$ files that are not stored have to be served by the central server. Therefore,
\begin{align*}
\EE\big[\widetilde{R}_{z_{>1}}^*\big]&\geq \int_{M+1}^{n} \Bigg(1-\Big(1-\frac{p_1}{i^\beta}\Big)^{\widetilde m}\Bigg)di\\
&= \int_{M+1}^{n} \frac{{\widetilde m}p_1}{i^\beta} (1 +\oo(1)) di\\
&=\frac{{\widetilde m}p_1}{\beta -1}\Bigg[\frac{1}{(M)^{\beta -1}}-\frac{1}{n^{\beta -1}}\Bigg](1 +\oo(1))\\
&\geq \frac{{\widetilde m}p_1}{\beta -1}\Bigg[\frac{1}{(M+1)^{\beta -1}}\Bigg](1 +\oo(1))\\
&=\Omega\Big(m^{\big(1-\mu(\beta-1)\big)}\Big).
\end{align*}

\textbf{Case 2:} $M=n$

We use Proposition \ref{prop:converse} to prove this result. Recall that the optimal solution to $\text{O}^*$ has the following structure: $\exists$ $i_{\min} \geq 1$ and $i_{\max} \leq n$, such that,
\begin{eqnarray*}
	x_{i} = \begin{cases}
		1, & \text{  if } i_{\text{min}} < {{i}} < i_{\text{max}}, \\
		f_1, & \text{if } i =i_{\text{min}} , \\
		f_2, & \text{if } i =i_{\text{max}} , \\
		0, & \text{otherwise},
	\end{cases}
\end{eqnarray*}
where, $0\leq f_1, f_2 \leq 1$.

Let $\widetilde i=\Big\lceil\big(\frac{\widetilde mp_1}{2}\big)^{\frac{1}{\beta}}\Big\rceil$. {By the definition of fractional knapsack problem}, we have that,
\begin{align}
f_1\frac{\widetilde{m}p_1}{i_{\min}^{\beta}}+\sum_{i=i_{\min}+1}^{\widetilde i}\big\lfloor{\widetilde mp_i}\big\rfloor+&\sum_{i=\widetilde i+1}^{i_{\max}-1}1+f_2 = M \nonumber \\
\implies f_1\frac{\widetilde{m}p_1}{i_{\min}^{\beta}}+\int_{i_{\min}+1}^{\widetilde i}{\widetilde mp_i}di+&i_{max}-2-2\widetilde i +f_2\leq M  \hspace{1.5in} \nonumber \\
\implies \int_{i_{\min}+1}^{\widetilde i}{\widetilde mp_i}di+i_{max}&\leq M+3\widetilde{i}-f_1\frac{\widetilde{m}p_1}{i_{\min}^{\beta}}.\nonumber 
\end{align}
\begin{align}
\therefore i_{max}\leq& M+3\widetilde{i}-f_1\frac{\widetilde{m}p_1}{i_{\min}^{\beta}}\nonumber\\
&+ \frac{\widetilde mp_1}{(\beta-1)}\bigg[-(i_{\min}+1)^{(-\beta+1)}+{\widetilde i}^{(-\beta+1)}\bigg]\nonumber  \\
=&M+3\bigg\lceil\Big(\frac{\widetilde mp_1}{2}\Big)^{\frac{1}{\beta}}\bigg\rceil-f_1\frac{\widetilde{m}p_1}{i_{\min}^{\beta}}+ \frac{\widetilde mp_1}{(\beta-1)}\nonumber\\
\Bigg[&-(i_{\min}+1)^{(-\beta+1)}+\bigg\lceil\Big(\frac{\widetilde mp_1}{2}\Big)^{\frac{1}{\beta}}\bigg\rceil^{(-\beta+1)}\Bigg].\nonumber
\end{align}

Let ${i_{\min}=m^{\alpha}}$. Recall that the fractional knapsack solution has $i_{\min}\leq \widetilde{i}$. Hence, $\alpha \leq \frac{1}{\beta}$.

If $\alpha < \frac{1}{\beta}$:

\begin{align}
 \text{ }i_{\max}&=M(1-\oo(1)).\label{eq:03010}\\
 &\text{Since, $M=n$} \nonumber\\
\frac{i_{\max}}{n}&=1-c_4m^{-\alpha \beta+\alpha}(1-\oo(1)). \label{eq:03020} 
\end{align}

{Let $\EE[R_1^*]$ denotes the expected number of contents, which are requested at least once and are more popular than Content $i_{\min}$. }

By Lemma \ref{lemma:Zipf_d} part (a),
{$$\EE[R_1^*]=m^{\alpha} \hspace{0.1in}$$ }

Let $\EE[R_2^*]$ denotes the expected number of contents, which are requested at least once and are less popular than Content $i_{\max}$.

\begin{align}
\EE[R_{2}^*]&\geq \int_{i_{max}+1}^{n}\frac{\widetilde mp_1}{i^\beta}di\nonumber \\
&=\frac{\widetilde mp_1}{\beta -1}\bigg[\frac{1}{(i_{max}+1)^{\beta -1}}-\frac{1}{n^{\beta -1}}\bigg]\nonumber \\
&=\frac{\widetilde mp_1}{(\beta -1)(i_{max}+1)^{\beta -1}}\bigg[1-\Big(\frac{i_{max}+1}{n}\Big)^{\beta -1}\bigg]\nonumber\\
&\text{From (\ref{eq:03010}) and (\ref{eq:03020})}\nonumber\\
&=c_5m^{1-\mu(\beta-1)-\alpha \beta+\alpha}(1+\oo(1)). \nonumber
\end{align}

\begin{align}
\EE\big[\widetilde{R}_{z_{>1}}^*\big]&\geq \EE[R_{1}^*]+\EE[R_{2}^*]\nonumber \\
&=m^\alpha+c_6m^{1-\mu(\beta-1)-\alpha \beta+\alpha}(1+\oo(1)). \nonumber
\end{align}

Therefore,
\begin{align}
 \EE\big[\widetilde{R}_{z_{>1}}^*\big] & \geq\Omega\bigg(m^{\max\{ \alpha, {1-\mu(\beta-1)-\alpha \beta+\alpha}\}}\bigg) \nonumber \\
 & = \Omega\Big(m^{\frac{1-\mu(\beta-1)}{\beta}}\Big).\nonumber
\end{align}
If $\alpha=\frac{1}{\beta}:$

In this case, if $i=\oo(m^{\frac{1}{\beta}})$ then Content $i$ is not stored in the caches. From Lemma \ref{lemma:Zipf_d} Part (a), all these files are requested at least once. Hence, $\forall$ $\delta>0$,

$$\EE\big[{\widetilde{R}_{z>1}}^*\big] \geq \Omega\Big(m^{\frac{1}{\beta}-\delta}\Big).$$

Therefore,
\begin{align}
\EE\big[\widetilde{R}_{z_{>1}}^*\big] & \geq \Omega\Big(m^{\frac{1-\mu(\beta-1)}{\beta}}\Big). \nonumber
\end{align}

\textbf{Case 3:} $M\geq\left(1+\epsilon\right)n,\text{ } \epsilon>0$\\

$0$ is trivial lower bound.

\end{proof}\\


\subsection{Proof of Theorem \ref{thm:ach_het}}
We use the following lemmas in the proof of Theorem \ref{thm:ach_het}. 

{Lemma \ref{lemma:serve_all_requests} states that with high probability,  all the  requests for files stored by the Knapsack  are served by the caches.}\\

\begin{lemma}
	\label{lemma:serve_all_requests}
	Let $\mathcal{R} = \{i: x_i = 1\}$, where $x_i$ is the solution of the fraction knapsack problem solved in Knapsack Storage: Part 1. Let $E_2$ be the event that if the top (largest) $\omega(m^{2-\beta+\delta})$ caches, for any $\delta>0$,  have the same storage size, the Match Least Popular policy matches all requests for all contents in $R$ to caches. Then we have
	$$\PP(E_2) = 1-\OO(n e^{-(\log m)^{2}}).$$  
\end{lemma}
\begin{proof}
	Before going into the proof we classify our caches into 2 types.\\
	\textit{Top caches:} which contain at least one file  from the set [$2:n_2$],\\
	\textit{Bottom caches:} which does not contain files from [$2:n_2$].  \\
	
	if the Knapsack Storage policy decides to store Content $i$, it stores it on $w_i$ caches. 
	\begin{eqnarray}\label{total_memory}
	\sum_{i=2}^{n_2} x_i w_i &\leq& \sum_{i=2}^{n_1} \bigg\lceil \bigg(1 + \dfrac{p_1}{2}\bigg) mp_i \bigg\rceil + \sum_{i=n_1+1}^{n_2} \lceil 4 p_1 (\log m)^2 \rceil \nonumber\\
	&\leq& \bigg(1 + \dfrac{p_1}{2}\bigg)  m(1-p_1)  \nonumber \\
	&& +  (4 p_1 (\log m)^2 + 1) m^{\frac{1+\delta}{\beta}} \nonumber\\
	&\leq& (1-c_7)m \hspace{2cm}\text{ for some } c_7>0.
	\end{eqnarray}
	
	Since the MLP policy matches requests starting from the least popular contents, we first focus on requests for contents less popular than Content $n_2$. These less popular contents are stored $\big\lceil{\frac{1}{\delta}+1}\big\rceil$ times. \\
	
	\textit{Bottom caches:} 
	
	Let $m^{\digamma}$ top (largest) caches have equal memory. Excluding File 1, the cumulative popularity of files of a cache belong to Bottom caches is
	\begin{align*}
	\PP_{bcp}&\leq \frac{p_1}{\Big(m^\frac{1+\delta}{\beta}\Big)^\beta}+\frac{p_1}{\Big(m^\frac{1+\delta}{\beta}+m\Big)^\beta}\\
	&\hspace{1.3in}+\sum_{i=1}^{n}\frac{p_1}{\Big(m^\frac{1+\delta}{\beta}+m+im^\digamma\Big)^\beta}\\
	&\leq\OO\bigg(\frac{p_1}{m^{1+\delta}}+\frac{p_1m^{1-\digamma}}{m^\beta}+\frac{p_1}{m^{\epsilon\beta}}\sum_{i=m^{1-\digamma}}^{n}\frac{1}{i^\beta}\bigg)\\
	&=\frac{p_1}{m^{1+\delta}}+\frac{p_1m^{1-\digamma}}{m^\beta}+\frac{p_1}{(\beta-1)m^{\digamma\beta}}\\
		&\hspace{1.5in}\Bigg[\frac{1}{m^{(1-\digamma)(\beta-1)}}-\frac{1}{n^{\beta-1}}\Bigg]\\
	&\text{If $\digamma\geq 2-\beta+\delta,\hspace{0.5in}$ } 
	\PP_{bcp} =\OO\Big(\frac{1}{m^{1+\delta}}\Big).
	\end{align*} 
	
    Let $X_{b,i}$ denotes the number of requests arrived for files in the Cache $i$  except File $1$, Then $\EE[X_{b,i}]= \OO\Big(\frac{1}{m^{\delta}}\Big).$
	By the Chernoff bound in Lemma \ref{lemma:chernoff}, $\PP(X_{b,i}\geq 1)=\OO\Big(\frac{1}{m^{\delta}}\Big).$ For $j\in[n_2, n]$, File $j$ is stored on $\lceil\frac{1}{\delta}+1\rceil$ caches. Hence, by the union bound, the probability that File $j$ unmatched is  $\OO(m^{-\delta}).$\\

	\textit{Top caches:}
	
	Excluding files belong to $[1:n_2]$, the cumulative popularity of files of a cache belong to Top caches is
	\begin{align*}
	\PP_{tcp}&\leq \frac{p_1}{(c_8m)^\beta}+\sum_{i=1}^{n}\frac{p_1}{\Big(c_8m+im^\epsilon\Big)^\beta} \hspace{0.5in} c_8\geq c_7\\
	&{=\OO\Bigg(\sum_{i=0}^{m^{1-\digamma}}\frac{p_1}{(c_8m)^\beta}+\sum_{i=m^{1-\epsilon}}^{n}\frac{p_1}{\Big(fm+im^\digamma\Big)^\beta}\Bigg)}\\
	&=\OO\Big(\frac{p_1m^{1-\digamma}}{(c_8m)^\beta}+\frac{p_1}{m^{\digamma\beta}}\sum_{i=m^{1-\digamma}}^{n}\frac{1}{i^\beta}\Big)\\
	&{=\frac{p_1m^{1-\digamma}}{(c_8m)^\beta}+\frac{p_1}{(\beta-1)m^{\digamma\beta}}\Bigg[\frac{1}{m^{(1-\digamma)(\beta-1)}}-\frac{1}{n^{\beta-1}}\Bigg]}\\
	&\text{If $\digamma\geq 2-\beta+\delta$,} \hspace{0.5in}
	\PP_{tcp} =\OO\Big(\frac{1}{m^{1+\delta}}\Big).
	\end{align*} 
	
	Let $X_{t_1,i}$ denotes the number of requests arrived for files in the Cache $i$  except for files ranked between  $1$ and $n_2$, Then $\EE[X_{t_1,i}]= \OO\Big(\frac{1}{m^{\delta}}\Big).$
	By the Chernoff bound in Lemma \ref{lemma:chernoff}, $\PP(X_{t_1,i}\geq 1)=\OO\Big(\frac{1}{m^{\delta}}\Big).$ For $j\in[n_2, n]$, File $j$ is stored on $\lceil\frac{1}{\delta}+1\rceil$ caches. Hence, by the union bound, the probability that File $j$ such that $j\in [n_2,n]$ unmatched is  $\OO(m^{-\delta}).$
	
	Next, we focus on files ranked between $2$ and $n_2$. Let $D_i$ be the set of caches storing File $i$ for $2 \leq i \leq n_2$. From Equation (\ref{total_memory}) if files are stored according to Knapsack Storage: Part 2, each cache stores at most one file with index $i$ such that $2 \leq i \leq n_2$. 
	
	We now focus on requests for contents ranked between $n_1$ and $n_2$.  Each file is stored on $|D_i|=\lceil 4 p_1 (\log m)^2 \rceil$ caches.  Let $X_{t_2,i}$ be the number of requests for these caches from files ranked greater than $n_2$, Then $\EE[X_{t_2,i}]= \OO\Big(\lceil 4 p_1 (\log m)^2 \rceil\frac{1}{m^{\delta}}\Big).$
	By the Chernoff bound in Lemma \ref{lemma:chernoff}, $\PP\Big(X_{t_2,i}\geq 2 p_1 (\log m)^2\Big)=\OO\bigg(\Big(\frac{1}{m^{\delta}}\Big)^{2 p_1 (\log m)^2}\bigg).$ From Lemma \ref{lemma:Zipf_d}, each content is requested not more than  $2 p_1 (\log m)^2$ times w.h.p. Hence, by the union bound, with probability $\geq 1-\OO\bigg(m\Big(\frac{1}{m^{\delta}}\Big)^{2 p_1 (\log m)^2}\bigg)$, all requests for contents in $R$ ranked between $n_1$ and $n_2$ are matched to caches by Match Least Popular. 
	
	We now focus on requests for contents ranked between $2$ and $n_1$.  File $i$ is stored on $|D_i|=\lceil (1+\frac{p_1}{2})mp_i \rceil$ caches.  Let $X_{t_3,i}$ be the number of requests for these caches from files ranked greater than $n_2$, Then $\EE[X_{t_3,i}]= \OO\Big(\lceil (1+\frac{p_1}{2})mp_i \rceil\frac{1}{m^{\delta}}\Big).$
	By the Chernoff bound in Lemma \ref{lemma:chernoff}, $\PP\Big(X_{t_3,i}\geq \frac{p_1}{4}mp_i\Big)=\OO\bigg(\Big(\frac{1}{m^{\delta}}\Big)^{ (\log m)^2}\bigg).$ From Lemma \ref{lemma:Zipf_d}, each content is requested not more than  $(1+\frac{p_1}{4})mp_i$ times w.h.p. Hence, by the union bound, with probability $\geq 1-\OO\bigg(m\Big(\frac{1}{m^{\delta}}\Big)^{(\log m)^2}\bigg)$, all requests for contents in $R$ ranked between $2$ and $n_1$ are matched to caches by Match Least Popular.
	
		We now focus on the requests for Content 1. Recall that if the Knapsack Storage policy decides to cache Content 1, it is stored on all $m$ caches. Since the total number of requests in a batch is $m$, even if all requests for contents ranked lower than 1 are matched to caches, the remaining caches can be used to serve all the requests for Content 1.

\end{proof}\\


The next lemma evaluates the performance of the Knapsack Store + Match Least Popular (KS+MLP) policy for the case where the content popularity follows the Zipf distribution.

\begin{lemma}
	\label{lemma:performance:our_policy}
	Consider a distributed content delivery system satisfying Assumption \ref{ass:system},  and the top (largest) $\omega(m^{2-\beta+\delta})$ caches, for any $\delta>0$ have the same storage size. 
	Let $R_{\text{KS+MLP}}$ be the transmission rate for the Knapsack Storage + Match Least Popular policy when content popularity follows the Zipf distribution with Zipf parameter $\beta>1$.  
	Then  for $m$ large enough, we have
	\begin{eqnarray*}
		\EE[R_{\text{KS+MLP}}] &\leq& \sum_{i \notin R} 1-\bigg(1-\frac{p_1}{i^{\beta}}\bigg)^{\widetilde{m}} + \OO(mn e^{-(\log m)^{2}}),
	\end{eqnarray*}
	where $p_1 = \big(\sum_{i=1}^n i^{-\beta}\big)^{-1}$, $\mathcal{R} = \{i: x_i = 1\}$, such that $x_i$ is the solution of the fraction knapsack problem solved in Knapsack Storage: Part 1.
\end{lemma}
\begin{proof}
	From Lemma \ref{lemma:serve_all_requests}, we know that, for $m$ large enough, with probability $\geq 1-\OO\big(ne^{-(\log m)^2}\big)$, all requests for the contents stored in caches by the KS+MLP policy are matched to caches. Let $\widetilde{n}$ be the number of contents not in $\mathcal{R}$ (i.e., not cached by the KS+MLP policy) that are requested at least once in a given time-slot. Therefore, 
	$$\EE[\tilde{n}] = \sum_{i \notin \mathcal{R}} 1-(1-p_i)^{\widetilde{m}},$$
	and
	\begin{eqnarray*}
		\EE[R_{\text{KS+MLP}}] &\leq& \EE[\tilde{n}] P(E_2) + m (1-P(E_2)) \\
		&\leq& \EE[\tilde{n}] + \OO(mn e^{-(\log m)^{2}}). 
	\end{eqnarray*}
	
\end{proof}\\


\begin{proof} (Theorem \ref{thm:ach_het})
	
	
	\textbf{Case 1:} $M\leq(1-\epsilon)n,\text{ } 0<\epsilon<1$
	
	 From Lemma \ref{lemma:serve_all_requests}, if we store File $i$ on $w_i$ caches and employ the Knapsack Storage Policy: Part 2, all the requests for File $i$ in a batch are served locally with high probability. Consider an alternative storage policy which starts storing files $2, 3, \dots$, each on $w_i$ caches respectively until the cache memory is exhausted. This policy stores Files 2, 3, \dots, $\frac{M-(1-\frac{p_1}{2})m}{\lceil{\frac{1}{\delta}+1}\rceil}$. Let $\mathbb{E}[R]$ be the expected transmission rate for this policy.
	By the definition of fractional knapsack problem, $\EE\big[\widetilde{R}_{z_{>1}}^{\text{KS}}\big] \leq \EE[R]$. Therefore,
%
%
	\begin{align*}
	\mathbb{E}\big[\widetilde{R}_{z_{>1}}^{\text{KS}}\big]&\leq 1+ \int_{\frac{M-(1-\frac{p_1}{2})m}{\lceil{\frac{1}{\delta}+1}\rceil}}^{n}1-(1-p_i)^{\widetilde{m}} di\\
	&= 1+ \int_{\frac{M-(1-\frac{p_1}{2})m}{\lceil{\frac{1}{\delta}+1}\rceil}}^{n}\widetilde{m}p_i (1+\oo(1))di\\
	&\leq 1+\frac{2\widetilde{m}p_1}{\beta -1} \Bigg[\frac{1}{\Big(\frac{M-(1-\frac{p_1}{2})m}{\lceil{\frac{1}{\delta}+1}\rceil}\Big)^{\beta-1}}-\frac{1}{(n)^{\beta-1}}\Bigg]\\
	&\approx \OO\Big(m^{1-\mu(\beta-1)}\Big)
	\end{align*}

	\textbf{Case 2:} $M=n$
	
	Let Knapsack solution decides to store files from $i_{min}+1$ to $i_{max}$.

	\begin{align}
	&\int_{i_{\min}}^{\frac{m^{\frac{1}{\beta}}}{(\log m)^\frac{2}{\beta}}}\bigg\lceil \Big(1 + \frac{p_1}{2}\Big) {\widetilde{m}p_i}\bigg\rceil di+\big\lceil {4 p_1(\log m)^2} \big\rceil n^{\frac{1+\delta}{\beta}}\nonumber\\
	&\hspace{2 in}-n^{\frac{1+\delta}{\beta}}+i_{\max} \geq M\nonumber\\
	&\int_{i_{\min}}^{\frac{m^{\frac{1}{\beta}}}{(\log m)^\frac{2}{\beta}}}\Big(1 + \frac{p_1}{2}\Big) {\widetilde{m}p_i}di+{4 p_1(\log m)^2}n^{\frac{1+\delta}{\beta}}+i_{\max} \geq M\nonumber
	\end{align}
	\begin{align}
	i_{\max}\geq &M- {4p_1(\log m)^2}n^{\frac{1+\delta}{\beta}}+\nonumber\\
	&\Big(1+\frac{p_1}{2}\Big)\frac{mp_1}{(\beta-1)}\Bigg[\Bigg(\frac{m^{\frac{1}{\beta}}}{(\log m)^\frac{2}{\beta}}\Bigg)^{1-\beta}-i_{\min}^{(1-\beta)}\Bigg] \nonumber \\
	\nonumber
	\end{align}
	
	Let $i_{\min}=m^{\alpha}$ such that $\alpha<\frac{1}{\beta}$ and substitute it in the above equation, we get
	\begin{align}
	i_{\max}&=M(1-\oo(1))\label{eq:imax2}\\
	\frac{i_{\max}}{n}&=1-c_9{m^{-\alpha(\beta-1)}}(1-\oo(1)) \label{eq:imax/n2}
	\end{align}
	\begin{align*}
	\mathbb{E}\big[\widetilde{R}_{z_{>1}}^{\text{KS}}\big]&\leq m^\alpha+\int_{i_{\max}}^{n} 1-(1-p_i)^{\widetilde{m}} di\\
	&= m^\alpha+\int_{i_{\max}}^{n} \widetilde{m}p_i (1+\oo(1)) di\\
	&= m^\alpha+\frac{2\widetilde{m}p_1}{\beta -1}\Bigg[\frac{1}{i_{\max}^{\beta-1}}-\frac{1}{n^{\beta-1}}\Bigg]\\
	&= m^\alpha+\frac{\widetilde{m}p_1}{(\beta -1)i_{\max}^{\beta-1}}\Bigg[1-\Big(\frac{i_{\max}}{n}\Big)^{\beta-1}\Bigg]\\
	&\text{From (\ref{eq:imax2}) and (\ref{eq:imax/n2})}\\
	&=\OO(m^{\alpha})+\OO\Big(m^{1-\mu(\beta-1)-\alpha(\beta-1)}\Big)
	\end{align*} 
	Optimize over $\alpha$, we will get
	\begin{align*}
	\mathbb{E}\big[\widetilde{R}_{z_{>1}}^{\text{KS}}\big]=\OO\Big(n^\frac{1-\mu(\beta-1)}{\beta}\Big)
	\end{align*}

{	\textbf{Case 3:} {$M\geq(1+\epsilon)n,\text{ } \epsilon>0$}
	
	If we store from Content $t$ to Content  $n$, the total memory required is lesser than 
	\begin{align}\label{eqn:finalmemory}
	&\int_{t}^{\frac{m^{\frac{1}{\beta}}}{(\log m)^\frac{2}{\beta}}}\Big(1 + \frac{p_1}{2}\Big) {\widetilde{m}p_i} di+\frac{4 p_1(\log m)^2}{a}  n^{\frac{1+\delta}{\beta}}+n \nonumber \\
	&\leq \Big(1+\frac{p_1}{2}\Big)\frac{mp_1}{a(\beta-1)}\Bigg[t^{(1-\beta)}-\Bigg(\frac{m^{\frac{1}{\beta}}}{(\log m)^\frac{2}{\beta}}\Bigg)^{1-\beta}\Bigg] \nonumber\\
	&\hspace{1.4in}+\frac{4 p_1(\log m)^2}{a} n^{\frac{1+\delta}{\beta}}+n.
	\end{align}
	
	Hence, $\exists$ $t$,  a constant such that, equation (\ref{eqn:finalmemory}) is less than $M$.
	
	Hence, 
	$$\mathbb{E}\big[\widetilde{R}_{z_{>1}}^{\text{KS}}\big]\leq t=\Theta(1).$$}

\end{proof}

\subsection{Proof of Theorem \ref{thm:oneuserpercache}}
\begin{proof}
 In our system, assume that the cumulative memory of $m-c_{10}m^{1-\mu(\beta-1)-\delta}$ caches (say low memory caches)  is $\leq (1-\epsilon)n.$ Consider a new system with $m+c_{10}m^{1-\mu(\beta-1)-\delta}$ caches, such that $m$ caches are similar to our system and the remaining $c_{10}m^{1-\mu(\beta-1)-\delta}$ caches (say new caches), have $\frac{n}{m}$ units of memory each.
 From Theorem \ref{thm:conv_het} Case 1, new caches + low memory caches can serve at most $\widetilde{m}-m^{1-\mu(\beta-1)}$ requests, and the remaining $c_{10}m^{1-\mu(\beta-1)-\delta}$ caches can serve $c_{10}m^{1-\mu(\beta-1)-\delta}$ requests. Hence, $$\mathbb{E}[\widetilde{R}_{z_{>1}}^*]\geq \Omega\big(m^{1-\mu(\beta-1)}\big).$$
\end{proof}